\documentclass[journal]{IEEEtran}

\usepackage[latin1]{inputenc}
\usepackage{graphicx, graphics, balance}
\usepackage{amssymb, amsmath, amsfonts, amsthm}
\usepackage{algorithm, algorithmic, float, xspace}
\usepackage[usenames,dvipsnames]{color}
\usepackage{cite, epsfig, tikz, standalone}
\usepackage[font=small]{caption}

\newtheorem{theorem}{Theorem}
\newtheorem{lemma}{Lemma}

\newcommand{\PP}{\mathbb{P}}
\newcommand{\E}{\mathbb{E}}
\newcommand{\R}{\mathbb{R}}
\newcommand{\la}{\lambda}
\newcommand{\La}{\mathcal{L}}
\newcommand{\ts}{\theta}
\newcommand{\s}{\sigma}
\newcommand{\al}{\alpha}
\newcommand{\sinr}{{\rm SINR}}

\allowdisplaybreaks

\begin{document}

\title{Successive Interference Cancellation\\ in Bipolar Ad Hoc Networks with SWIPT}

\author{Constantinos Psomas, \IEEEmembership{Member, IEEE}, and Ioannis Krikidis, \IEEEmembership{Senior Member, IEEE}\vspace{-4mm}
\thanks{C. Psomas and I. Krikidis are with the  KIOS Research Center for Intelligent Systems and Networks, University of Cyprus, Cyprus (e-mail: \{psomas, krikidis\}@ucy.ac.cy).}
\thanks{This work was supported by the Research Promotion Foundation, Cyprus, under the project FUPLEX with Pr. No. CY-IL/0114/02.}}

\maketitle

\begin{abstract}
Successive interference cancellation (SIC) is based on the idea that some interfering signals may be strong enough to decode in order to be removed from the aggregate received signal and thus boost performance. In this letter, we study the SIC technique from a simultaneous wireless information and power transfer (SWIPT) standpoint. We consider a bipolar ad hoc network and evaluate the impact of SIC on the SWIPT performance for the power splitting technique. Theoretical and numerical results show that our proposed approach can achieve significant energy gains and under certain scenarios the average harvested energy converges to its upper bound.
\end{abstract}

\begin{keywords}
SWIPT, power splitting, successive interference cancellation, stochastic geometry. \vspace{-2mm}
\end{keywords}

\section{Introduction}

\IEEEPARstart{E}{nergy} harvesting from radio frequency (RF) signals has been proposed as an alternative to conventional methods. In this context, simultaneous wireless information and power transfer (SWIPT) is a new communication paradigm, where a wireless device can obtain both information and energy from the received RF signals \cite{RUI}. Although information theoretic studies ideally assume that a receiver is able to decode information and harvest energy independently from the same signal \cite{GRO}, this approach is not feasible due to practical limitations. In order to achieve SWIPT, the received RF signal needs to be partitioned into two parts, one for information transfer and another for energy harvesting, which can be done in the time, power, antenna or space domain \cite{KRI}.

As the effect of interference plays a significant role in both information and energy transfer in wireless systems, SWIPT has been studied extensively in multi-user systems \cite{JP,NZ,KRI2,KHAN}. In this letter, we study SWIPT in a multi-user large-scale system and also take into consideration the spatial randomness of the users \cite{KRI2,KHAN}. The work in \cite{KRI2}, considers a SWIPT network with multiple source-destination pairs and the performance is derived under two protocols: non-cooperative and cooperative. In \cite{KHAN}, the impact of millimeter waves and blockages on SWIPT is investigated for cellular networks. The performance of SWIPT, in a large-scale network, is also highly affected by the multi-user interference which produces a strong trade-off between information decoding and harvested energy; that is, high levels of interference increase the harvested energy but decrease the performance of information decoding and vice versa. From the information decoding standpoint, many interference cancellation methods have been proposed in order to improve its performance; a well-known method is successive interference cancellation (SIC) \cite{SIC}, \cite{SIC2}. In SIC, the receiver attempts to decode the strongest interfering signals and, if successful, they are effectively removed thus increasing the signal-to-interference-plus-noise-ratio (SINR). The performance gains from the employment of SIC have been shown for cellular networks \cite{SIC} but also ad hoc networks \cite{SIC2}. However, SIC has never been used as a technique to achieve SWIPT.

In this letter, we study the employment of the SIC method in a bipolar ad hoc network where the receivers employ SWIPT with the power splitting (PS) technique. We show how each receiver can utilize SIC in order to boost the wireless power transfer without affecting the information decoding. Both analytical and numerical results are presented for the coverage probability and the average harvested energy before and after the employment of SIC. Our results demonstrate that SIC is significantly beneficial for SWIPT systems and we show that in certain scenarios the harvested energy converges to its upper bound, i.e. the case where all the power from the received signal is used for harvesting.\vspace{-1mm}

\begin{figure}[t]\centering
  \includegraphics[width=0.6\linewidth]{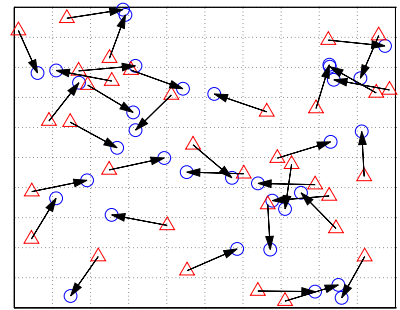}
  \caption{A snapshot of a bipolar ad hoc network; the receivers and transmitters are depicted by circles and triangles, respectively.}\label{model}\vspace{-5mm}
\end{figure}

\section{System model}

\subsection{Network \& channel model}
We consider a large-scale bipolar ad hoc wireless network consisting of a random number of transmitter-receiver pairs \cite[Ch. 12]{BACC}, \cite[Sec. 5.3]{HAE}. The bipolar model is ideal for modeling ultra-dense ad hoc networks, where devices such as smartphones, tablets, sensors, etc. will be pairwise connected and operating at the same frequency band; these devices can be considered as randomly spread. The transmitters form a homogeneous Poisson point process (PPP) $\Phi = \left\{x_i \in \R^2\right\}$ with $i \geq 1$ of density $\la$, where $x_i$ denotes the location of the $i$-th transmitter. Each transmitter $x_i$ has a unique receiver at a distance $d_0$ in some random direction. We consider a receiver located at the origin and its associated transmitter $x_0$. Based on Slivnyak's Theorem \cite{HAE}, we perform our analysis for this typical receiver but our results hold for any receiver in the network. Fig. \ref{model} schematically shows the considered network model. All nodes are equipped with a single antenna. The time is considered to be slotted and in each time slot all the transmitters are active without any coordination or scheduling. All wireless links suffer from both small-scale block fading and large-scale path-loss effects. The fading is considered to be Rayleigh distributed so the power of the channel fading is an exponential random variable with unit variance. We denote by $h_i$ the channel coefficient for the link between the $i$-th transmitter and the typical receiver. Moreover, all wireless links exhibit additive white Gaussian noise (AWGN) with variance $\sigma^2$. The path-loss model assumes the received power is proportional to $(1 + d_i^\al)^{-1}$ where $d_i$ is the Euclidean distance from the origin to the $i$-th transmitter and $\alpha > 2$ is the path-loss exponent. The probability density function of the distance $r$ from the origin to the $n$-th nearest neighbour is given by \cite[Sec. 2.9.1]{HAE},
\begin{align}f(r,n) = \frac{2(\pi \la)^n}{\Gamma(n)} r^{2n-1} e^{-\pi \la r^2}.\label{pdf}\end{align}

\subsection{Joint wireless information and power transfer model}
It is assumed that the transmitters have a continuous power supply, such as a battery or the power grid, and transmit with the same power $P_t$. Each receiver has SWIPT capabilities and employs the PS method such that the received signal is split into two parts: one is converted to a baseband signal for information decoding and the other is directed to the rectenna for energy harvesting and storage. Note that the time switching method is a special case of PS with a binary power splitting parameter and so it is not considered explicitly. Let $0 < v \leq 1$ denote the PS parameter for each receiver, i.e. $100v\%$ of the received power is used for decoding. The additional circuit noise formed during the RF to baseband conversion phase is modeled as an AWGN with zero mean and variance $\s^2_C$. Therefore, the signal-to-interference-plus-noise ratio (SINR) at the typical receiver can be written as
\begin{align}
\sinr_0 = \frac{v P_t h_0 \tau^{-1}}{v(\s^2 + P_t I_0) + \s^2_C},\label{sinr}
\end{align}
where $\tau = 1 + d_0^\al$ and $I_0 = \sum_{x_i \in \Phi} h_i (1+d_i^\al)^{-1}$, $i > 0$ is the aggregate interference at the typical receiver. The RF energy harvesting is a long term operation and is expressed in terms of average harvested energy \cite{KRI2}. As $100(1 - v)\%$ of the received energy is used for rectification, the average energy harvested at the typical receiver is expressed as
\begin{align}
E(v) = \zeta \,\E\left[(1-v)P_t\left(h_0 \tau^{-1} + I_0\right)\right],\label{avg_en}
\end{align}
where $0 < \zeta \leq 1$ denotes the conversion efficiency from the RF signal to DC voltage. Note that any RF energy harvesting from the AWGN noise is considered to be negligible.

\subsection{Successive interference cancellation model}\label{sic}

The receivers are assumed to apply the SIC\footnote{We assume an ideal SIC where interfering signals are cancelled effectively, i.e. no residual interference exists. Furthermore, power consumption related to SIC is not considered.} technique to cancel the $n$ strongest interfering signals, $n \in \mathbb{N}$. The central idea behind SIC is to attempt to decode the strongest interfering signals such that after a successful attempt, they can be removed from the aggregate signal at the receiver, which will provide a higher SINR for the useful signal. Specifically, the main steps of the SIC protocol are as follows \cite{SIC2}. The receiver tries to decode the useful signal from its associated transmitter; if successful, no SIC is employed. Otherwise, the receiver attempts to decode the strongest interfering signal and remove it from the received signal. The SINR is then re-evaluated and the receiver re-attempts to decode the useful signal. If the decoding of the useful signal is still unsuccessful, the receiver proceeds to decode the next strongest interfering signal and remove it from the received signal. This procedure repeats up to $n$ times, during which the receiver will either manage to decode the useful signal or after the $n$-th attempt it will be in outage. By assuming that the order statistics of the received interfering signal power are determined by the distance, i.e. $h_i (1+d_i^\al)^{-1} > h_j (1+d_j^\al)^{-1}$ for $i < j$, the above states can be expressed mathematically as
\begin{align*}
S_0 &: \frac{v P_t h_0 \tau^{-1}}{v(\s^2 + P_t I_0) + \s^2_C} \geq \ts, \\
S_n &: \left(\bigcap_{j=0}^{n-1}\frac{v P_t h_0 \tau^{-1}}{v(\s^2 + P_t I_j) + \s^2_C} < \ts\right) \\
&~~~ \bigcap \left(\bigcap_{j=1}^n \frac{v P_t h_j (1+d_j^\al)^{-1}}{v(\s^2 + P_t(I_j + h_0 \tau^{-1})) + \s^2_C} \geq \ts\right) \\
&~~~ \bigcap \left(\frac{v P_t h_0 \tau^{-1}}{v(\s^2 + P_t I_n) + \s^2_C} \geq \ts\right),
\end{align*}
where $I_j = \sum_{x_i \in \Phi} h_i (1+d_i^\al)^{-1}$, $i > j$. In contrast to existing works on SIC which ideally assume that the useful signal can be ignored during the decoding of interfering signals \cite{SIC}, \cite{SIC2}, in our approach, we consider it as an extra interference term.

\section{SWIPT with SIC}

In this section, we present the main results of our proposed technique. We derive analytically the probabilities for the main states of SIC given above and the expected harvested energy at the receiver. The performance of SIC is evaluated in terms of the coverage probability (also called success probability), that is, the probability that the $\sinr$ is above the target threshold $\theta$, i.e. $\mathbb{P}[\text{SINR}>\theta]$ \cite{SIC2}. In what follows, we assume that the interfering terms $I_i$, $i \geq 0$ are mutually independent \cite{SIC}.

\begin{lemma}\label{lemma1}
The coverage probability of a receiver which has not applied SIC is
\begin{align}
  \Pi_{\rm NC}(v) &= \exp\left(-\frac{\ts \tau}{P_t}\left(\s^2 + \frac{\s^2_{C}}{v}\right)\right) \nonumber\\ 
  &\times\exp\left(-\frac{2 \pi^2 \la}{\al} \ts \tau (1 + \ts\tau)^{\left(\frac{2}{\al}-1\right)} \csc\left(\frac{2 \pi}{\al}\right)\right).
\end{align}
\end{lemma}

\begin{proof}
Since $h_0$ is an exponential random variable with unit variance and $\Pi_{\rm NC}(v) = \PP\left[\sinr_0 > \ts\right]$, we have
\begin{align}
\PP\left[\sinr_0 > \ts\right] = \exp\left(-\frac{\ts \tau}{P_t}\left(\s^2 + \frac{\s^2_{C}}{v}\right)\right) \La_{I_0}\left(\ts \tau\right),\label{prob1}
\end{align}
where $\La_{I_0}(\ts \tau)$ is the Laplace transform of the random variable $I_0$ evaluated at $\ts \tau$ and is derived as follows
\begin{align}
&\La_{I_0}\left(\ts \tau\right) = \E\left[\exp\left(-\ts \tau \sum_{x_i \in \Phi} \frac{h_i}{1+d_i^\al} \right)\right] \nonumber \\
&= \E\left[\prod_{x_i \in \Phi} \exp\left(-\ts \tau \frac{h_i}{1+d_i^\al} \right)\right] \nonumber \\
&\stackrel{(a)}{=} \exp\left(-2 \pi \la \int_0^\infty \left(1-\frac{1}{1+ \ts \tau(1+x^\al)^{-1}}\right)x dx\right) \label{lap2} \\
&\stackrel{(b)}{=} \exp\left(-\frac{2}{\al} \la \pi^2 \ts \tau (1 + \ts \tau)^{\left(\frac{2}{\al}-1\right)} \csc\left(\frac{2 \pi}{\al}\right)\right),\label{lap1}
\end{align}
where $(a)$ follows from the probability generating functional of a PPP \cite{HAE} and from the moment generating function (MGF) of an exponential random variable since the $h_i$ are i.i.d.; $(b)$ by using the transformation $u \leftarrow x^2$ and the Mellin transform of the resulting expression \cite[17.41]{GRAD}. The result follows by substituting \eqref{lap1} into \eqref{prob1}.
\end{proof}

\begin{lemma}
The coverage probability of a receiver attempting to decode the $n$-th interferer is
\begin{align}
&\Pi_{\rm D}(v,n) = \int_0^\infty \frac{f(r,n)}{1 + \ts \xi \tau^{-1}} \exp\left(-\frac{\ts \xi}{P_t}\left(\s^2 + \frac{\s^2_{C}}{v}\right)\right) \nonumber\\
&\times\exp\left(\!-\frac{2\pi \la \ts \xi r^{2-\al}}{\al-2} {}_2F_{1}\left(\! 1, 1-\frac{2}{\al}; 2-\frac{2}{\al}; -\frac{1 + \ts \xi}{r^{\al}}\right)\right) dr,
\end{align}
where $\xi = (1 + r^\al)$ and $f(r,n)$ is given by \eqref{pdf} and ${}_2F_{1}(\cdot,\cdot;\cdot;\cdot)$ denotes the Gauss hypergeometric function.\vspace{-1mm}
\end{lemma}

\begin{proof}
The proof follows similarly to the one of Lemma \ref{lemma1} with the only difference that the coverage probability is conditioned on the distance $r$ to the $n$-th neighbour. Furthermore, in this case the useful signal is treated as an interfering signal. Therefore, the SINR in this case is
\begin{align}
\sinr_{\rm D} = \frac{v P_t h_n (1+r^\al)^{-1}}{v(\s^2 + P_t(I_n + h_0 \tau^{-1})) + \s^2_C},
\end{align}
where $I_n = \sum_{x_i \in \Phi} h_i (1+d_i^\al)^{-1}$, $i > n$. Therefore, we have
\begin{align}
&\!\!\!\Pi_{\rm D}(v,n) = \PP[\sinr_{\rm D} > \ts ~|~ r] \nonumber \\
&\!\!\!\!=\!\int_0^\infty f(r,n) \exp\left(-\frac{\ts (1 + r^\al)}{P_t}\left(\s^2 + \frac{\s^2_{C}}{v}\right)\right) \nonumber\\
&\quad~~\times \E\left[\exp\left(-\ts (1 + r^\al)h_0\tau^{-1}\right)\right]\La_{I_n}\left(\ts (1+r^\al)\right)dr \nonumber\\
&\!\!\!\!=\!\int_0^\infty \!\frac{f(r,n)}{1 + \ts \xi \tau^{-1}} \exp\left(\!-\frac{\ts \xi}{P_t}\left(\s^2 + \frac{\s^2_{C}}{v}\right)\right) \La_{I_n}\left(\ts \xi\right)dr,\! \label{prob2}
\end{align}
which follows from the MGF of an exponential random variable and by setting $\xi = (1 + r^\al)$. The Laplace transform is evaluated as in Lemma \ref{lemma1} by setting the limits of the integral in \eqref{lap2} from $r$ to $\infty$ and using the expressions in \cite[3.24]{GRAD}.
\end{proof}

\begin{lemma}
The coverage probability of a receiver which has successfully canceled $n$ interferers is
\begin{align}
&\Pi_{\rm C}(v,n) = \exp\left(-\frac{\ts \tau}{P_t}\left(\s^2 + \frac{\s^2_{C}}{v}\right)\right) \int_0^\infty f(r,k) \nonumber \\
&\times\!\exp\left(\!-\frac{2\pi \la \ts \tau r^{2 - \al}}{\al-2} {}_2F_{1}\left(\!1, 1-\frac{2}{\al}; 2 - \frac{2}{\al}; -\frac{1 + \ts \tau}{r^\al}\right)\right) dr,
\end{align}
where $f(r,k)$ is given by \eqref{pdf} and ${}_2F_{1}(\cdot,\cdot;\cdot;\cdot)$ denotes the Gauss hypergeometric function.
\end{lemma}

\begin{proof}
The SINR in this case is as in \eqref{sinr} but with $I_0 = \sum_{x_i \in \Phi} h_i (1+d_i^\al)^{-1}$, $i > n$. The proof then follows the same steps as above.
\end{proof}
Note that when $n = 0$, i.e. no interferers have been cancelled, $\Pi_{\rm C}(v,n)$ is equal to the coverage probability when no SIC is applied, that is $\Pi_{\rm C}(v,0) = \Pi_{\rm NC}(v)$. We now provide the two main theorems of this letter, the coverage probability of a receiver with SIC capabilities and the receiver's average harvested energy.

\begin{theorem}
The coverage probability of a receiver which attempts to cancel up to $n$ interferers is
\begin{align}
&\!\!\!\!\Pi_{\rm SIC}(v,n) = \Pi_{\rm NC}(v) \nonumber\\
&\!\!\!\!+ \sum_{i=1}^n \left(\prod_{j=0}^{i-1} \left(1-\Pi_{\rm C}(v,j)\right)\right) \left(\prod_{j=0}^{i} \Pi_{\rm D}(v,j)\right) \Pi_{\rm C}(v,i).
\end{align}
\end{theorem}

\begin{proof}
For analytical tractability, we assume the interference terms are independent. Then, the proof follows directly from the SIC protocol described in Section \ref{sic}.
\end{proof}

\begin{theorem}
The average harvested energy by the typical receiver is given by\vspace{-1mm}
\begin{align}
E(v) = \zeta (1 - v) P_t \left(\tau^{-1} + \frac{2}{\al} \pi^2 \la \csc\left(\frac{2 \pi}{\al}\right)\right).
\end{align}
\end{theorem}

\begin{proof}
Similarly to \cite{KRI2}, using Campbell's theorem for the expectation of a sum over a point process \cite{HAE} and the fact that $\E(h_i)=1$ for all $i$, we get from \eqref{avg_en},
\begin{align}
E(v) = \zeta(1-v)P_t\left(\tau^{-1} + \E[I_0]\right),
\end{align}
where,\vspace{-2mm}
\begin{align}
\E[I_0] &= \E\left[\sum_{x_i \in \Phi} (1+d_i^\al)^{-1}\right] \nonumber \\
&= 2 \pi \la \int_0^\infty \frac{r}{1+r^\al} dr = \frac{2}{\al} \pi^2 \la \csc\left(\frac{2 \pi}{\al}\right), \nonumber
\end{align}
which follows from using the transformation $u \leftarrow r^2$ and \cite[3.241]{GRAD}.
\end{proof}

\begin{figure*}[t]
  \begin{minipage}{0.305\linewidth}\vspace{4mm}
    \includegraphics[width=\linewidth]{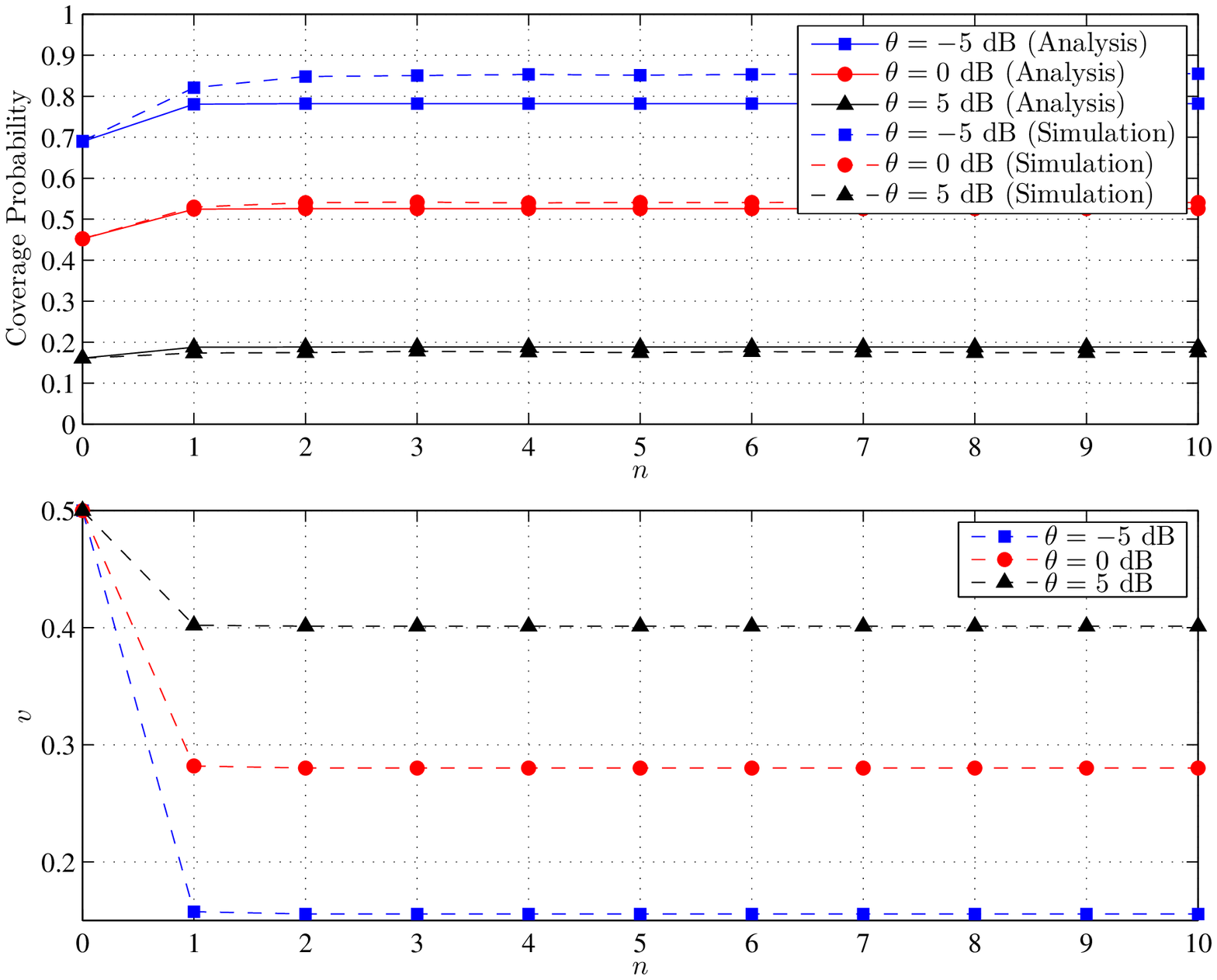}\vspace{1mm}
    \captionof{figure}{Coverage probability and $v$ versus the number of cancelled interferers $n$ for $\ts = \{-5, 0, 5\}$ dB.}\label{fig1}
  \end{minipage}\hfill
  \begin{minipage}{0.33\linewidth}
    \includegraphics[width=\linewidth]{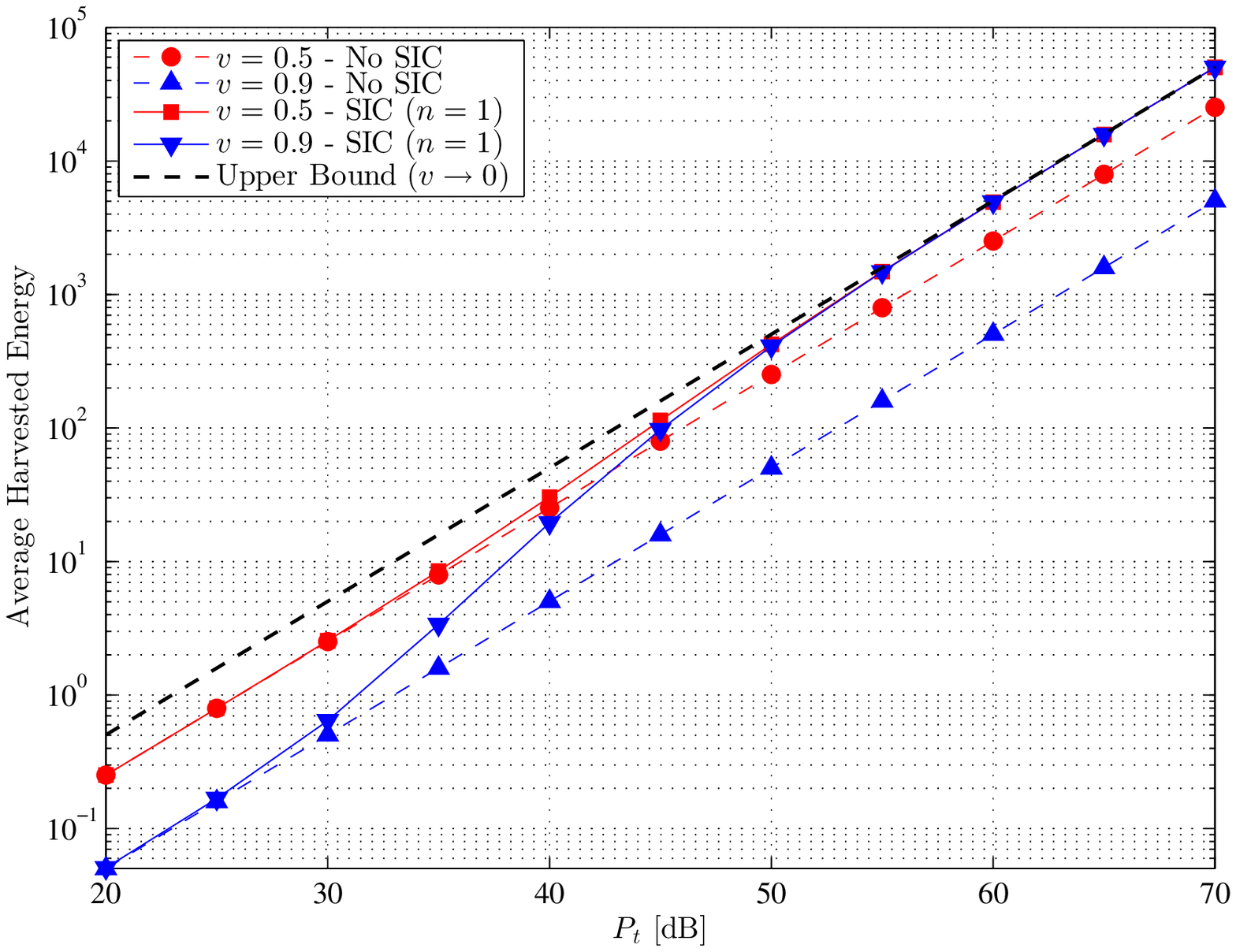}
    \captionof{figure}{Average harvested energy versus the transmit power $P_t$ for $v = \{0.5,0.9\}$.}\label{fig2}
  \end{minipage}\hfill
  \begin{minipage}{0.33\linewidth}
    \includegraphics[width=\linewidth]{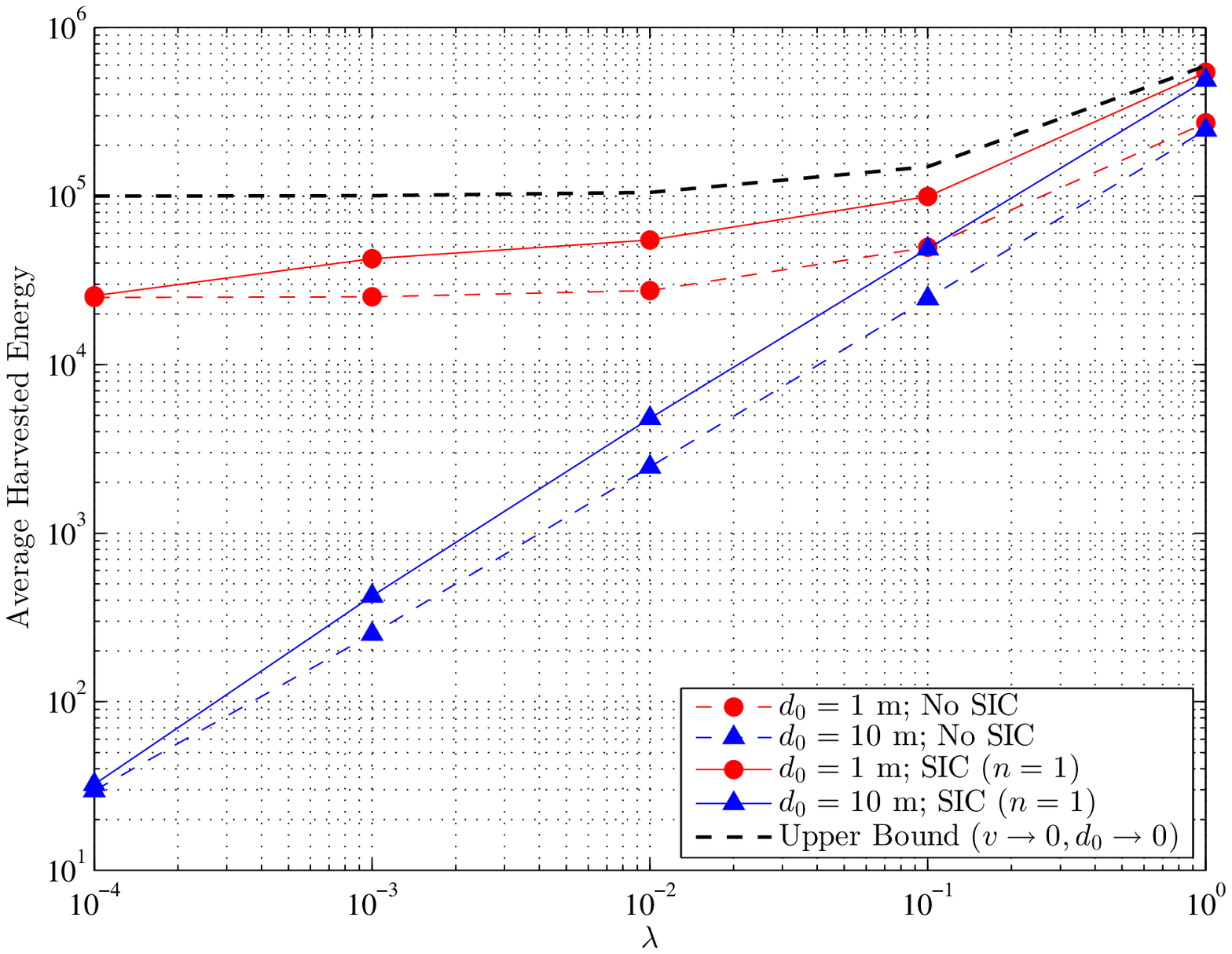}
    \captionof{figure}{\small Average harvested energy versus the density $\la$ for $d_0 = \{1, 10\}$ m.}\label{fig3}
  \end{minipage}\vspace{-3mm}
\end{figure*}

We now look at how SIC can be exploited in order to increase the average harvested energy. As the employment of SIC provides an improvement to the coverage probability, we can reduce the PS parameter in such a way so that the achieved performance is still as good as the case where SIC is not applied. The benefit from this method is that as the PS parameter decreases, more power is provided to the harvesting operation thus increasing the average harvested energy. In other words, we would like to compute the following
\begin{align}
& \max_{v} ~ E(v) \nonumber\\
& \text{subject to}\qquad 0 < v\leq 1, \nonumber\\
& \qquad\qquad\qquad \Pi_{\rm SIC}(v,n) \geq \eta,
\end{align}
where $\eta \geq \Pi_{\rm NC}(v)$ is the target outage probability. In order to maximize the average harvested energy at the receiver, it suffices to minimize $v$ subject to the above constraints. This can be found by finding the root of the equation $\Pi_{\rm SIC}(v,n) = \eta$. Due to the nonlinearity of the expression, the root can be found using numerical methods such as the bisection method.

\section{Numerical Results}

We provide both simulation and theoretical results in order to validate our proposed technique and demonstrate the potential benefits from its implementation. Unless otherwise stated, we use $\la = 10^{-3}$, $P_t = 50$ dB, $d_0 = 10$ m, $\ts = -5$ dB, $\s^2 = 1$, $\s^2_C = 1$, $v = 0.5$, $\zeta = 1$ and $\al = 4$. Furthermore, we set $\eta = \Pi_{\rm NC}(v)$.

Fig. \ref{fig1} depicts the coverage probability and the optimal PS parameter $v$ with respect to $n$, the maximum number of cancelled interfering signals. Note that the difference between the analytical and simulation results is due to the independence assumption between the interfering terms. However, it's clear that we can still capture the behaviour of the model and also this difference decreases as $\ts$ increases. Next, observe the validity of our proposed technique. For $n = 1$, the coverage probability increases for all values of $\ts$ and at the same time $v$ decreases, so more power is provided to the harvesting operation. The coverage gain for $n = 1$ is greater for lower values of $\ts$ which explains the larger reduction in $v$. Furthermore, for $n > 1$, SIC makes no significant difference to the coverage probability which is the expected behaviour, so $v$ remains unchanged as well.

Fig. \ref{fig2} shows the impact of the transmit power $P_t$ and SIC on the average harvested energy for different values of the PS parameter $v$. Clearly, the employment of our method provides significant energy harvesting gains. This is because an increase in $P_t$ provides better quality signal, which results to more power for harvesting. For $v = 0.9$, the harvesting gains are noticeable from $P_t = 25$ dB whereas for $v = 0.5$ from $P_t = 40$ dB. This is again due to the fact that there is more than enough power at the receiver to satisfy the performance threshold $\ts$ so the adjustment of $v$ starts from smaller $P_t$. Moreover, for large values of $P_t$, the energy harvested with the proposed method converges to the upper bound ($v \rightarrow 0$).

Finally, Fig. \ref{fig3} illustrates how the density of the network $\la$ and the distance $d_0$ effect the average harvested energy. As expected, the smaller the distance $d_0$ and the greater the density $\la$, the larger the average harvested energy is. Also, our proposed method provides energy harvesting gains which increase with $\la$ since in a more dense network the probability of decoding the strongest interfering signal is higher. Thus SIC increases the coverage probability which in turn provides energy harvesting benefits. For large values of $\la$ the harvested energy converges to the same value regardless of $d_0$ since in this case the interfering signals dominate the network and so the energy harvested from the direct link is insignificant. Lastly, in Fig. \ref{fig3} we also plot the upper bound of the average harvested energy for $v \rightarrow 0$ and $d_0 \rightarrow 0$; it's clear that our proposed method converges to the upper bound for large values of $\la$.

\section{Conclusion}
In this letter, we studied SIC in a bipolar ad hoc network from a SWIPT point of view. We showed that SIC can significantly increase wireless power transfer without affecting the information decoding. The potential energy harvesting gains from our proposed approach were illustrated with analytical and numerical results. Future extensions of this method include the investigation of the non-ideal case of SIC by taking into account residual interference and power consumption related to its implementation.

\end{document}